\documentclass[journal, draftclsnofoot, onecolumn]{IEEEtran}
\usepackage{cite}

\ifCLASSINFOpdf
\else
\fi
  \usepackage[pdftex]{graphicx}

\usepackage[cmex10]{amsmath}
\usepackage{amsfonts}
\usepackage{amssymb,framed}
\usepackage{amsthm}

\newtheorem{thm}{Theorem}[section]
\newtheorem{cor}[thm]{Corollary}
\newtheorem{lem}[thm]{Lemma}

\newtheorem{defn}[thm]{Definition}

\newtheorem{rmk}[thm]{Remark}

\DeclareMathOperator*{\supp}{supp}
\DeclareMathOperator*{\argmin}{arg\,min}

\usepackage{algorithm}
\usepackage{algorithmic}

\usepackage{enumerate}

\begin{document}
%
\title{New Coherence and RIP Analysis for Weak Orthogonal Matching Pursuit}
%
%
%

\author{Mingrui~Yang,~\IEEEmembership{Member,~IEEE,}
        and~Frank~de~Hoog
\thanks{Mingrui Yang is with CSIRO Computational Informatics, Pullenvale, QLD 4069 Australia (email: mingrui.yang@csiro.au).}
\thanks{Frank de Hoog is with CSIRO Computational Informatics, Acton, ACT 2061 Australia (email: frank.dehoog@csiro.au).}
}

\maketitle

\begin{abstract}
In this paper we define a new coherence index, named the \emph{global 2-coherence}, of a given
dictionary and study its relationship with the traditional mutual
coherence and the restricted isometry constant. By exploring this
relationship, we obtain more general results on sparse signal reconstruction using
greedy algorithms in the compressive sensing (CS) framework. In
particular, we obtain an improved bound over the best known results
on the restricted isometry constant for
successful recovery of sparse signals using orthogonal matching pursuit (OMP).
\end{abstract}


%
\IEEEpeerreviewmaketitle

\section{Introduction}
\label{sec:intro}

 

Compressive sensing (CS) \cite{CandesRombergTao:06, Do} is a newly developed and fast growing field
of research. It provides a new sampling scheme that breaks the
traditional Shannon-Nyquist sampling rate
\cite{Shannon1949:samplingtheorem} given that the signal of interest
is sparse in a certain basis or tight frame. More specifically, for a vector $a \in \mathbb{R}^d$, let
$\|a\|_0$ denote the $\ell_0$ ``norm'' of $a$, which counts the number of
nonzero entries of $a$. We say $a$ is \emph{$k$-sparse} if $\|a\|_0
\le k$. CS has established conditions for finding the unique sparse solution of the
following $\ell_0$ minimization problem
\begin{align}\label{eqn:l0_minimization}
  \min_a \|a\|_0 \mbox{ subject to } f = \Phi a,
\end{align}
where $\Phi \in \mathbb{R}^{n\times d}$ ($n \ll d$) and $f\in
\mathbb{R}^n$. To ensure that the $k$-sparse solution is unique, we
need the following restricted isometry property introduced by Candes
and Tao in \cite{Candes2005:LP}.
\begin{defn}[Restricted Isometry Property (RIP)]
  A matrix $\Phi$ satisfies the \emph{restricted isometry property}
  of order $k$ with the \emph{restricted isometry constant (RIC)} $\delta_k$
  if $\delta_k \in (0,1)$ is the smallest constant such that
  \begin{align}
    (1-\delta_k)\|a\|_2^2 \le \|\Phi a\|_2^2 \le (1+\delta_k)\|a\|_2^2
  \end{align}
  holds for all $k$-sparse signal $a$.
\end{defn}

It has been shown in \cite{Candes2005:LP} that if $\delta_{2k} < 1$,
then the $\ell_0$ minimization problem \eqref{eqn:l0_minimization} has
a unique $k$-sparse solution. However, solving an $\ell_0$
minimization problem is in general NP-hard. One solution to this
problem is to relax the $\ell_0$ ``norm'' to the
$\ell_1$ norm. Candes has shown in \cite{Candes:08} if
$\delta_{2k} < \sqrt{2} - 1$, then $\ell_1$ minimization is equivalent to
$\ell_0$ minimization. 
Another alternative is to use heuristic greedy algorithms
to approximate the solution of the $\ell_0$ minimization problem
. Orthogonal matching pursuit (OMP)
is one of the simplest and most popular algorithms of this type. For
the analysis of greedy algorithms, the metric chosen for the sensing
matrix are usually coherence indices rather than the RIC.

For simplicity, from now on, we always assume that the columns of the
matrix (dictionary) $\Phi$ are normalized such that for any column
$\phi\in\Phi$, $\|\phi\|_2 = 1$.
\begin{defn}\label{def:mutualcoherence}
  The \emph{mutual coherence} $M(\Phi)$ of a matrix $\Phi$ is defined by
  \begin{align}
    M(\Phi) := \max_{\substack{\phi_i, \phi_j \in \mathbf{\Phi} \\ i \neq
        j}} | \langle \phi_i, \phi_j  \rangle |,
  \end{align}
  where $\langle \cdot, \cdot \rangle$ represents the usual inner product.
\end{defn}
It has been shown that if $(2k-1)M < 1$, then OMP can recover every
$k$-sparse signal exactly in $k$ iterations \cite{DET}. Recently, researchers have started to investigate the
performance of OMP using RIP. Davenport and
Wakin~\cite{5550495} have proved that $\delta_{k+1} <
\frac{1}{3\sqrt{k}}$ is sufficient for OMP to recover any $k$-sparse
signal in $k$ iterations. Mo and Shen \cite{MoShen:12} improve the bound to $\delta_{k+1} <
\frac{1}{1+\sqrt{k}}$. They also give an
example that OMP fails to recover a $k$-sparse signal in $k$ steps
when $\delta_{k+1} = \frac{1}{\sqrt{k}}$. This leaves a question if
their bound can be further improved.

It is then natural to examine the
relationship between the mutual coherence $M$ and the RIC $\delta_k$,
since the bound for $M$ is already sharp. However, approaching this
directly was not fruitful and this motivated us to define a new coherence
index, namely the 
\emph{global 2-coherence}, and
establish a bridge connecting the mutual coherence, the global 2-coherence,
and the RIC. Then by using this connection, we analyze
the performance of weak orthogonal matching pursuit (WOMP), a weak
version of OMP. In particular, we extend the results given
in~\cite{yang2014} to show that $\delta_k +
\sqrt{k}\delta_{k+1} < 1$ is sufficient for OMP to recover any
$k$-sparse signal in $k$ iterations, which provides an improved bound
over the best known result given in~\cite{MoShen:12} and confirms that
it is not yet optimal. As mentioned above, the
results presented in this paper is an extension of~\cite{yang2014},
where we introduced a new algorithm to CS, called orthogonal matching
pursuit with thresholding (OMPT), and showed its reconstruction
stability and robustness.


\section{Global 2-Coherence}

We first define a new coherence index, the \emph{global 2-coherence}, $\nu_k(\Phi)$ for a given
dictionary $\Phi$. Then based on this new coherence index, we establish the
connections among the coherence indices and the RIC $\delta_k$.

\begin{defn}\label{def:newcoherence}
  Denote by $[d]$ the index set $\{1,2,\ldots,d\}$. The global 2-coherence of a dictionary $\Phi \in \mathbb{R}^{n\times d}$ is defined as
  \begin{equation}\label{eqn:newcoherence}
    \nu_k(\Phi):=
    \max_{i\in[d]}\max_{\substack{\Lambda\subseteq[d]\setminus\{i\} \\ |\Lambda| \le k}}
    \left( \sum_{j\in\Lambda} \langle \phi_i,\phi_j \rangle^2
    \right)^{1/2},
  \end{equation}
  where $\phi_i$, $\phi_j$ are columns from the dictionary $\Phi$.
\end{defn}

The global 2-coherence $\nu_k (\Phi)$ defined above is a generalization of the mutual coherence defined in
Definition~\ref{def:mutualcoherence} and the coherence indices defined
in \cite{1337101,5361489}. In particular, when $k=1$, $\nu_1$ is
exactly the mutual coherence.






The following lemma describes the
relations among the mutual coherence $M$, the 2-coherence $\nu_k$, and
the restricted isometry constant $\delta_k$.

\begin{lem}\label{lem:relationship}
    For $k > 1$, we have
    \begin{equation}
        M \le \nu_{k-1} \le  \delta_{k} \le \sqrt{k-1}\nu_{k-1} \le (k-1)M.
    \end{equation}
\end{lem}


The next lemma is needed to proceed to our main results.

\begin{lem}\label{lem:bounds_noisy}
  Let $\Lambda \subset [d]$ with $|\Lambda| = k$. Let $f
  = \Phi a + w$ with $\supp(a) = \Lambda$ and $\|w\|_2 \le
  \epsilon$. In addition, assume that there exits $\Omega \subseteq
  \Lambda$ with $|\Omega| = m$, such that
  \begin{align*}
    \langle \Phi a, \phi_i \rangle = 0, \mbox{ for $i\in \Lambda \setminus
      \Omega$}.
  \end{align*}
  Then
  \begin{align*}
    \max_{i \in [d]\setminus\Lambda} |\langle f,\phi_i \rangle|
    &\le
    \nu_k \|a\|_2 + \epsilon, \\
    \max_{i \in \Lambda} |\langle f,\phi_i \rangle|
    &\ge
    \frac{\sqrt{1-\delta_k}}{\sqrt{m}}  \|\Phi a\|_2 - \epsilon.
  \end{align*}
\end{lem}


\section{Main Results}
We first begin with a well known greedy algorithm, the weak orthogonal
matching pursuit (WOMP), which was defined
in~\cite{Temlyakov:00}. Here we present a simple version in
Algorithm~\ref{alg:womp} where the weak parameter $\rho$ is a constant
for each iteration. Notice that when $\rho = 1$, WOMP becomes standard
OMP.

\begin{algorithm}
  \caption{Weak Orthogonal Matching Pursuit (WOMP)}
  \begin{algorithmic}[1]
    \STATE\textbf{Input:} weak parameter $\rho\in (0,1]$, dictionary $\Phi$,
    signal $f$, and the noise level $\epsilon$.
    \STATE\textbf{Initialization:} $r_0 := f$, $x_0:=0$, 
    $\Lambda_0:=\emptyset$, $s:=0$.
    \WHILE {$\|r_s\|_2 > \epsilon$}
      \STATE Find an index $i$ such that $$|\langle r_s,\phi_i\rangle|
      \ge \rho\cdot\max_{\phi}|\langle r_s,\phi\rangle|$$
      where $\phi$ is any column of $\Phi$;
      \STATE Update the support: $$\Lambda_{s+1} = \Lambda_s \cup \{i\};$$
      \STATE Update the estimate: $$x_{s+1} = \argmin_z \|f - \Phi_{\Lambda_{s+1}} z\|_2;$$
      \STATE Update the residual: $$r_{s+1} = f - \Phi_{\Lambda_{s+1}} x_{s+1};$$
      \STATE $s = s+1$;
    \ENDWHILE
    \STATE \textbf{Output:} If the algorithm is stopped after $k$
    iterations, then the output estimate $\hat{a}$ of $a$ is
    $\hat{a}_{\Lambda_k} = x_k$ and $\hat{a}_{\Lambda_k^C} = 0$.
  \end{algorithmic}
  \label{alg:womp}
\end{algorithm}

Let us consider the case where a sparse signal is
contaminated by a perturbation. Specifically, let $\Lambda \subset
[d]$ with $|\Lambda| = k$. We consider a signal $f = \Phi
a + w$, where $a\in\mathbb{R}^d$ with $\supp(a) = \Lambda$ and
$\|w\|_2 \le \epsilon$. 

\begin{thm}\label{thm:womp_noisy}
  Denote by $a_{min}$ the nonzero entry of $a$ with the least
  magnitude, and $\hat{a}_{womp}$ the recovered representation of $f$
  in $\Phi$ by WOMP after $k$ iterations. If
  \begin{align}
    \sqrt{k} \nu_k < \rho (1-\delta_k) \label{eqn:thm_womp_bound1_noisy}
  \end{align}
  and the noise level obeys
  \begin{align}
    \epsilon < \frac{\rho (1-\delta_k) - \sqrt{k}
      \nu_k}{1+\rho} |a_{min}|, \label{eqn:thm_womp_bound2_noisy}
  \end{align}
  then
  \begin{enumerate}[a)]
        \item $\hat{a}_{womp}$ has the correct sparsity pattern
            \begin{equation*}
                \supp(\hat{a}_{womp}) = \supp(a);
            \end{equation*}
         \item $\hat{a}_{womp}$ approximates the ideal noiseless representation
             \begin{equation}
                 \|\hat{a}_{womp} - a\|_2^2 \leq
                 \frac{\epsilon^2}{1-\delta_k}. \label{eqn:thm_womp_errorbound_noisy}
             \end{equation}
    \end{enumerate}
\end{thm}



From Lemma~\ref{lem:relationship}, it follows that
\begin{cor} \label{cor:womp}
  Let $f = \Phi a$ with $\|a\|_0 = k$. If one of the following
  conditions is satisfied,
  \begin{enumerate}[a)]
    \item $\sqrt{k} \delta_{k+1} < \rho
      (1-\delta_{k})$,\label{eqn:cor_womp_condition_delta}
    \item $\sqrt{k} \nu_k < \rho (1-\nu_{k-1} \sqrt{k-1})$, 
    \item $kM < \rho(1-(k-1)M)$,
  \end{enumerate}
  then, $a$ is the unique sparsest representation of $f$ and moreover,
  WOMP recovers $a$ exactly in $k$ iterations.
\end{cor}


The performance of WOMP decreases as $\rho$ decreases. Now if we set
$\rho = 1$ in WOMP, then we obtain immediately the following corollary for OMP.

\begin{cor}\label{cor:omp_delta}
  Let $f = \Phi \alpha$ with $\|a\|_0 = k$. If
  \begin{align}
    \delta_k + \sqrt{k} \delta_{k+1} < 1, \label{eqn:cor_omp_condition_delta}
  \end{align}
  then, $a$ is the unique sparsest representation of $f$ and moreover,
  OMP recovers $a$ exactly in $k$ iterations.
\end{cor}

\begin{rmk}
  The condition in~\eqref{eqn:cor_omp_condition_delta} gives an improved bound on the
restricted isometry constant compared to the bound obtained in
\cite{MoShen:12} for OMP for successful recovery
after $k$ iterations, where the bound was
$\delta_{k+1} < \frac{1}{\sqrt{k} + 1}$.
\end{rmk}

\section{Conclusion}
In this paper, we have introduced a new generalized coherence index,
the \emph{global 2-coherence}, and established two
connections among the mutual coherence, the global 2-coherence, and the
restricted isometry constant. Based on these relations,
we analyzed the performance of WOMP as well as OMP for their recovery
ability of sparse representations in both ideal noiseless and noisy
cases. In particular, for the noiseless case, we showed an improved
bound over the best known results on the restricted isometry constant
for successful recovery using OMP.


%



\appendix

\begin{proof}[Proof of Lemma~\ref{lem:relationship}]
	It is easy to show that $\nu_k$ increases with $k$ while
        $\frac{\nu_k}{\sqrt{k}}$ decreases with $k$. Therefore, the first and
        the last relations follow immediately. 
	
	We now prove the second inequality. 
	\begin{align*}
          \nu_{k-1}(\Phi)
          &=
          \max_{i\in[d]} \max_{\substack{\Lambda\subseteq[d]\backslash\{i\}\\|\Lambda| \le k-1}} \left( \sum_{j\in\Lambda} \langle \phi_i, \phi_j \rangle^2 \right)^{\frac{1}{2}} \\
          &=
          \max_{\substack{\Lambda\subseteq[d] \\ |\Lambda| \le k}} \max_{i\in\Lambda} \left( \sum_{j\in\Lambda\backslash\{i\}} \langle \phi_i,\phi_j \rangle^2 \right)^{\frac{1}{2}} \\
          &=
          \max_{\substack{\Lambda\subseteq[d] \\ |\Lambda| \le k}} \|
          \Phi_\Lambda^T \Phi_\Lambda - I \|_{\infty, 2},
        \end{align*}
        where $\Phi_{\Lambda} \in \mathbb{R}^{n\times |\Lambda|}$ is a
        submatrix of $\Phi$ with columns indexed in $\Lambda$.
              
	On the other hand, according to Proposition~2.5 in \cite{Rauhut:10}, one has
	\begin{align*}
		\delta_k &=  \max_{\substack{\Lambda\subseteq[d] \\ |\Lambda| \le k}} \| \Phi_\Lambda^T \Phi_\Lambda - I \|_{2, 2} \\
		&\ge \max_{\substack{\Lambda\subseteq[d] \\ |\Lambda| \le k}} \| \Phi_\Lambda^T \Phi_\Lambda - I \|_{\infty, 2} \\
		&= \nu_{k-1}(\Phi),
	\end{align*}
	which completes the proof for the second inequality.
	
	Next we prove the third inequality. Consider the Gram matrix $G =\Phi_\Lambda^T \Phi_\Lambda$, where its entries $g_{ij} = \langle \phi_i,\phi_j \rangle$. Clearly its diagonal entries $g_{ii} = 1$. Then by the Gershgorin Circle Theorem, each eigenvalue $\lambda$ of $G$ is in at least one of the disks $\{z: | z-1 | \le R_i \}$, where $R_i = \sum_{\substack{j\in\Lambda \\ j\ne i}} |g_{ij}|$. Equivalently, we have
	\begin{align*}
		1-R_i \le \lambda \le 1+R_i
	\end{align*}
for some $i$.
	Therefore,
	\begin{align*}
		\delta_k &\le \max_i R_i 
                             = \max_i \sum_{\substack{j\in\Lambda \\ j\ne i}}
                  |g_{ij}| \\
                             &\le \max_i \sqrt{k-1}
                             \Big(\sum_{\substack{j\in\Lambda \\ j\ne i}}
                          |g_{ij}|^2\Big)^{\frac{1}{2}} \\
                             &\le \sqrt{k-1}\nu_{k-1}.
	\end{align*}
\end{proof}

\begin{proof}[Proof of Lemma~\ref{lem:bounds_noisy}]
  For $i \in [d]\setminus\Lambda$, we have
  \begin{align*}
    |\langle f,\phi_i \rangle|
    &=
    | \langle \Phi a + w, \phi_i \rangle| \\
    &\le
   |\langle \Phi a,\phi_i \rangle| + |\langle w,\phi_i \rangle|\\
   &\le
   \nu_k \|a\|_2 + \|w\|_2 \|\phi_i\|_2 \\
   &\le
   \nu_k \|a\|_2 + \epsilon.
 \end{align*}
	Taking maximum on both sides completes the proof of the first inequality.

	Now for $i \in \Lambda$,
	\begin{align*}
          \max_{i \in \Lambda} |\langle f,\phi_i \rangle|
          &=
          \max_{i \in \Lambda} |\langle \Phi a + w,\phi_i \rangle|  \\
          &\ge
          \max_{i \in \Lambda} |\langle \Phi a , \phi_i \rangle| -
          \max_{i \in \Lambda} |\langle w, \phi_i \rangle| \\
          &\ge
          \frac{\sqrt{1-\delta_k}}{\sqrt{m}} \|\Phi a\|_2 - \max_{i
            \in \Lambda} \|w\|_2 \|\phi_i\|_2 \\
          &\ge
          \frac{\sqrt{1-\delta_k}}{\sqrt{m}} \|\Phi a\|_2 - \epsilon.
        \end{align*}
        This completes the proof of the second inequality.
\end{proof}

\begin{proof}[Proof of Theorem~\ref{thm:womp_noisy}]
  First, we show that WOMP recovers the correct support of $a$.

  We start with the first iteration. Note that $r_0=f$. We need to show
  \begin{align}
    \max_{i \in [d]\setminus\Lambda} |\langle f, \phi_i \rangle| < \rho \max_{i
    \in \Lambda} |\langle f, \phi_i \rangle|. \label{eqn:proof_thm_womp_noisy_keyrelation}
  \end{align}

  By Lemma~\ref{lem:bounds_noisy}, we have
  \begin{align}
    \max_{i \in [d]\setminus\Lambda} |\langle f, \phi_i \rangle|
    &\le
    \nu_k \|a\|_2 + \epsilon, \label{eqn:proof_thm_womp_noisy_upperbound}
  \end{align}
  and
  \begin{align}
    \max_{i \in \Lambda} |\langle f,\phi_i \rangle|
    &\ge
    \frac{1-\delta_k}{\sqrt{k}} \|\Phi a\|_2 - \epsilon \notag \\
    &\ge
    \frac{1-\delta_k}{\sqrt{k}} \|a\|_2 - \epsilon. \label{eqn:proof_thm_womp_noisy_lowerbound}
  \end{align}

  Now since $\|a\|_2 \ge \sqrt{k} |a_{min}|$, by imposing conditions \eqref{eqn:thm_womp_bound1_noisy} and
  \eqref{eqn:thm_womp_bound2_noisy}, we get
  \begin{align*}
    \nu_k \|a\|_2 + \epsilon < \rho \left(
    \frac{1-\delta_k}{\sqrt{k}} \|a\|_2 - \epsilon \right),
  \end{align*}
  and relation~\eqref{eqn:proof_thm_womp_noisy_keyrelation} follows
  from the two bounds
  \eqref{eqn:proof_thm_womp_noisy_upperbound} and
  \eqref{eqn:proof_thm_womp_noisy_lowerbound}.
  Hence, WOMP only selects one atom from $\{\phi_i\}_{i\in\Lambda}$ in
  the first iteration.

  Now we argue that by repeatedly applying the above procedure, we are
  able to correctly recover the support of $a$. In fact, we have for
  the $s$-th iteration
  \begin{align*}
    r_s
    &=
    f - P_{\Lambda_s}(f) \\
    &=
    \Phi a + w -
    \left( P_{\Lambda_s}(\Phi a) +
      P_{\Lambda_s}(w) \right) \\
    &=
    (I - P_{\Lambda_s}) \Phi a + (I -
    P_{\Lambda_s}) w \\
    &=
    \Phi a_s + w_s
  \end{align*}
  where
  $$ \Phi a_s =  (I -
  P_{\Lambda_s}) \Phi a$$
  and
  $$w_s = (I -P_{\Lambda_s}) w.$$
  Therefore, $\langle \Phi a_s, \phi_i \rangle = 0$ for
  $i\in\Lambda_s$. Note that $(k-s)$ components of $a_s$ are the same as
  that of $a$. Then the result follows from the inequality
  $\|a_s\|_2 \ge \sqrt{k-s}|a_{min}|$ and
  Lemma~\ref{lem:bounds_noisy} for $s$-th iteration. In addition, the
  orthogonal projection step guarantees that the procedure will not
  repeat the atoms already chosen in previous iterations. Therefore,
  the correct support of the noiseless representation $a$ can be
  recovered exactly after $k$ iterations.

   Next, by following the idea of the proof of Theorem 5.1 in
   \cite{DET} and using the relation $\sigma_{min} \ge 1 - \delta_k$,
   where $\sigma_{min}$ denotes the smallest singular value of $\Phi$,
   we are able to prove the error bound~\eqref{eqn:thm_womp_errorbound_noisy}. 

\end{proof}



\ifCLASSOPTIONcaptionsoff
  \newpage
\fi



\bibliographystyle{IEEEtran}
\bibliography{ga}
\end{document}